\documentclass  {Domfundam}
\usepackage{amssymb}

\usepackage[T1]{fontenc}
\usepackage[utf8]{inputenc}

\usepackage{amscd}
\usepackage{amsmath}
\usepackage{textcomp}

\usepackage{graphicx}
\usepackage[usenames,dvipsnames,svgnames,table]{xcolor}
\usepackage{hyperref}
\usepackage{todonotes} 
\usepackage{listings}
\renewcommand{\implies}{\Rightarrow}
\renewcommand{\equiv}{\Leftrightarrow}
\newtheorem{metatwierdzenie}[theorem]{Metatheorem}

\newtheorem{pyta} {Question}

\author{
Andrzej Salwicki \\  Dombrova Research,\ \ \ salwicki at mimuw.edu.pl \\ Partyzantów 19, Łomianki PL-05-092}

\title{A new proof of   Euclid's algorithm}
 
\usepackage{tikz}
\usetikzlibrary{arrows,shadows,positioning}
\tikzset{
  frame/.style={
    rectangle, draw, 
    text width=12em, text centered,
    minimum height=3em,drop shadow,fill=lime!0,
    rounded corners,
  },
  line/.style={
    draw, -latex',rounded corners=3mm,
  }
}

\begin{document}
\label{tytul} 
 
 


\maketitle
\runninghead{A. Salwicki}{A new proof of Euclid's algorithm}

\begin{abstract}
Our main result is a new proof of correctness of Euclid's algorithm (\ref{E}). The proof is conducted in algorithmic theory of natural numbers $\mathcal{T}h_3$.    
 1\textdegree\  We are constructing a formula (\ref{H}) that expresses the halting property of the algorithm.
 2\textdegree\  Next, we analyze the structure of the algorithm (the proof makes no references to the computations of the algorithm). 
 3\textdegree\ Our proof makes use of inference rules of \textit{calculus of programs}  (i.e. algorithmic logic),  4\textdegree\ The only formulas accepted without proof are the axioms of the program calculus and axioms of algorithmic theory of natural numbers $\mathcal{T}h_3$. \\
 We consider also three other theories: first-order theory of natural numbers $\mathcal{T}h_0$, algorithmic theory of addition $\mathcal{T}h_1$, algorithmic theory of addition and multiplication $\mathcal{T}h_2$. It is easy to observe that $\mathcal{T}h_0 \subsetneq \mathcal{T}h_2$, $\mathcal{T}h_1 \subsetneq \mathcal{T}h_2$, $\mathcal{T}h_2 \subsetneq \mathcal{T}h_3$.  \\   

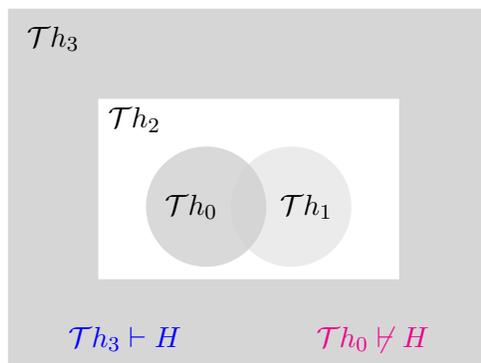
\begin{center}
 \begin{tikzpicture} 
 [scale=0.4]

\newcommand{\B}{(0,0) ++(45:2) circle (2)}
 
\fill[color=lightgray!65,even odd rule] (-5,-1) rectangle (5,5) (-8,-4) rectangle (8,8) ;
\fill[opacity=0.7,color=lightgray!85] (0,0) ++(135:2) circle (2) ;
\fill[opacity=0.4,color=lightgray!75] \B; 
\draw (-6.5,7 ) node{$\mathcal{T}h_3$};
\draw (0,-3) node{$\color{blue}{\mathcal{T}h_3\vdash H} \hspace*{1.8cm} \color{magenta}{\mathcal{T}h_0\not \vdash H}  $};
\draw (-3.8,4.3 ) node{$\mathcal{T}h_2$};
\draw (-0.5,0) ++(135:2) node{$\mathcal{T}h_0$};
\draw (0.5,0) ++(45:2) node{$\mathcal{T}h_1  $}; 
\end{tikzpicture}  \bigskip \\
\end{center}
We complete our result  showing that the  theorem on correctness of Euclid's algorithm can not be proved in any of  theories  $\mathcal{T}h_1, \mathcal{T}h_2, \mathcal{T}h_0$.
\end{abstract}
\tableofcontents
\section{Introduction}
History of Euclid's algorithm is over 2300 years long.
\begin{equation}\tag{E}\label{E}\underbrace{\left\{\begin{array}{l}\textbf{while }n \neq m \textbf{ do }\\ \quad \textbf{if }n>m   \textbf{ then }  n:=n-m   \textbf{ else }    m:=m-n   \textbf{ fi}\\ \textbf{ od} \end{array}\right\}}_{Euclid's\ algorithm} \end{equation}
An execution of  the algorithm, for initial values $n_0$ and $m_0$ results in the sequence of states. \\ 
If a state satisfies the condition $ n=m $ then the algorithm stops and $n$ is the greatest common divisor of $n_0$ and $ m_0 $.\\
 Euclid's algorithm is used in solutions of many   problems of programming and mathematics. It is applied  in various data structures. In some structures (e.g. the ring of integers, the field of rational numbers) the algorithm stops and in others (e.g. the field of real numbers, the planar geometry) its execution can be prolonged \textit{ad infinitum}. \\
The traditional proof of the theorem that the Euclid's algorithm halts and computes the greatest common divisor of two natural numbers is commonly accepted. Below you will find an analysis of this proof.  However, ...
\begin{center}
	being true does not mean being a proved theorem.
\end{center}
At present we remark that, the traditional  proof uses some assumptions and bases on the notion of execution of the algorithm. Therefore some \textit{semantics} is involved in the proof.\smallskip\\  
We are arguing that it is possible to present a purely syntactic proof of halting property of the algorithm.\\

To begin with, we need a formula that expresses the stopping property of Euclid's algorithm. In 1967 Erwin Engeler  c.f. \cite{EEnge}, found that, for every deterministic while program $M$ there is  an infinite disjunction (of open formulas) which does express the halting property of the program.
In the case of Euclid's algorithm the halting formula may look as follows
{\footnotesize \begin{equation}\nonumber \forall_{n\neq 0}\forall_{m\neq 0}\left(
\begin{array}{l}
n=m \,\lor\\
n=2m \lor 2n=m \,\lor \\
n=3m \lor 2n=3m \lor 3n=2m \lor m=3n\,\lor \\
\dots \lor \\
n=5m \lor 2n=5m \lor 3n=5m \lor 4n=5m \lor \\
\qquad 5n=m \lor 5n=2m \lor 5n=3m \lor 5n=4m \,\lor \\
\dots
\end{array} \right)
\end{equation}   }
We have appreciated the suggestion of Engeler and modified it in two ways.  First,  we 
  extended the language of first-order formulas by algorithms and algorithmic formulas, c.f. \cite{fal:as}, see also section II.3 of \cite{al:gm:as}. Next, we constructed algorithmic formulas that express various semantical properties of programs, e.g. correctness.  Our proposal is based on the remark that  the program itself is a finite expression. \\ The  formula (\ref{H}) below is the halting formula of the Euclid's algorithm.
{\footnotesize \begin{equation}\tag{H}\label{H}\underbrace{\begin{LARGE}\mbox{$\forall$}\end{LARGE}_{n \in N,n\neq 0}\,\begin{LARGE}\mbox{$\forall$}\end{LARGE}_{m \in N,m\neq 0}\left\{\begin{array}{l}\textbf{while }n \neq m \textbf{ do }\\ \ \textbf{if\, }n>m  \\ \textbf{\ \ \ then\, } n:=n-m  \\ \textbf{\ \ \ else\, }   m:=m-n  \\ \ \textbf{fi}\\ \textbf{od} \end{array}\right\}(n=m)}_{halting\ formula\ of\ Euclid's\ algorithm} \end{equation}  }
 Formula ($ H $)  reads informally: \textit{for every $n>0$ and $m>0$, algorithm (E) halts and the result of its computation satisfies formula } $n=m$. Hence, formula expresses the stop property of the program \ref{E}.   \\ 
Obviously, there are many formulas equivalent to the halting formula (H). Look at following formula
{\footnotesize \begin{equation}\tag{H'}\label{H'} \forall_{n \neq 0,m \neq 0 }\,  \begin{LARGE}\mbox{$\bigcup$}\end{LARGE} \left\{ \begin{array}{l}  \mathbf{if}\,n>m\,  \\ \ \ \mathbf{then}\, n:=n-m \, \\ \ \ \mathbf{else}\,   m:=m-n \, \\ \mathbf{fi} \end{array} \right\} (n=m)  \end{equation}} 
Formula (H') contains an \textit{iteration quantifier} \begin{Large}$\bigcup$\end{Large} and reads informally: \textit{for every $n>0$ and $m>0$, there exists an iteration of the program} \begin{center}
	\{\textbf{if} $n>m$ \textbf{then} $n:=n-m$ \textbf{else} $m:=m-n$ \textbf{fi} \}  
\end{center}
\textit{such that formula $n=m$ is satisfied afterwards}, i.e. the formula  \{\textbf{if} $n>m$ \textbf{then} $n:=n-m$ \textbf{else} $m:=m-n$ \textbf{ fi}\}$^i(n=m)$ holds .\\ You can convince yourself about correctness of the above proposition in section \ref{AL}. Note, there is no first-order formula that expresses the same meaning as formulas (H) or (H') do. \medskip\\
One can classify mathematical theories into three categories:
\begin{itemize}
	\item \textit{intuitive} -- this means that neither the language of theory nor its axioms are precisely described,
	\item \textit{axiomatized} -- since Euclid's \textit{Elements}, the majority of mathematical theories is \textit{axiomatized}, 
	\item \textit{formalized} -- this notion was introduced in foundations of mathematics when researchers encountered paradoxes such as Berry's paradox, cf.\cite{mm:hr:rs}.
\end{itemize}
    Till today, every one of known proofs of correctness of Euclid's algorithm is conducted in an intuitive number theory. Among others, it is assumed that the arguments are standard natural numbers. No axioms excluding non-standard elements accompany the proof. (Yes, we are aware that it is impossible in the frame of any first order theory. But are we limited to first-order logic?). The computations of the algorithm are studied, without mentioning that the notion of algorithm does not belong to the theory. The correspondence between the text of the algorithm and its computations is assumed in an intuitive way.  Nothing disturb us, we believe in the correctness of Euclid's algorithm. The proof is accepted by the vast majority of humans. \smallskip \\
However, is it an inter-subjective proof? Will it be accepted by computer? \\
It is worthwhile to think of creating such a proof that will be accepted by a proof-checker. 
\subsection{Do we, informaticians and  mathematicians need a new proof?}
 We are going to convince you that:
\begin{itemize}
 \item Proving properties of programs is like proving mathematical theorems. One needs: axioms, calculus of programs and well defined language. In other words, programmers, makers of specifications (of software),  verifiers of software    properties,   should accept algorithmic theories as the workplace. Also mathematicians will find many open problems in algorithmic theories of traditional data structures. For them also calculus of programs may be an interesting workplace.\\
 Such a theory contains classical theorems i.e. first-order formulas and algorithmic theorems as well. In the proofs of theorems on certain  programs we can use earlier theorems on other programs (and classical theorems also). In the proofs of some first-order theorems we can use some facts about programs. Section \ref{ATN} gives a flavor of such a theory, to be developed.
 \item Developers of algorithmic theories of numbers, of graphs, etc. should accept programs as ``first class citizens'' of the languages of these theories. Moreover, the language should contain formulas that express the semantical properties of programs. And, naturally, the reasoning should be done in a richer calculus of programs that contains the predicate calculus as its subset.
\end{itemize}
We believe that by constructing a new proof of correctness of Euclid's algorithm we shall gain a new insight into the nature of (algorithmic) theory of numbers. It is commonly accepted that algorithms play an important role in this theory. We need the tools adequate to the structure of analyzed texts. By this we mean that 1\textdegree\ the algorithms should be treated as ``\textit{first class citizens}'' of the theory, like the formulas are, 2\textdegree\ the semantical properties of algorithms should be expressed by the formulas, 3\textdegree\ these formulas should be the subject of studies having as aim their proof or a counterexample. \\ 
We expect that the proofs will be inter-subjective ones. It means, that everyone reading the proof will necessarily agree with the arguments. Finally, such a proof should be analyzable by a proof-checker \footnote{proof-checkers of algorithmic proofs do not   exists yet}. The fact of incompleteness of first-order theory of natural numbers should not be used as an indulgence for our laziness. \smallskip \\
The notion of halting formula may cause some doubts. We shall accept the following 
\begin{definition}
Let $K$ be a program.  Let   $\mathfrak{C}$ be a class of all (similar) structures that allow  to execute the program $ K $.    A \textit{halting formula} of the program $K$ is any formula $H(K)$ such that, for every data structure $\mathbb{A} \in \mathfrak{C}$,   the following conditions are equivalent
 \begin{itemize}
  \item [(\textit{i})] all executions of the program $K$ in data structure $\mathbb{A}$   are finite,
  \item[(\textit{ii})] the formula $H(K)$ is valid in structure $\mathbb{A}$.\hspace*{8.4cm} $\square$ 
 \end{itemize}
\end{definition}
We start with a suggestion: to analyze algorithms one should extend the language of mathematical theory. It means that the set of well formed expressions of a  language of a theory, should contain the properly defined set of algorithms (or programs) aside of sets of terms and formulas. See figure 1.  
In algorithmics the natural numbers and Euclid's algorithm play a significant role. (In the most used programming languages one encounters the structure of unsigned integers, i.e. natural numbers.) Moreover, in some programming languages we find a class\footnote{class is a kind of program module, cf. Appendix A} Int. The details of implementation can be hidden (even covered by patents). One may doubt, whether such a class is a proper model, and of which theory.  In the appendix A we show a class \textit{Cn} that satisfies the axioms of the theory of addition. We are also showing that for some objects of class \textit{Cn} as arguments,  the computations of Euclid's algorithm need not to be finite. 
It seems important that a theory in which we shall conduct the correctness proof of Euclid's algorithm will not have non-standard models, see sections \ref{firstPeano}, \ref{Th1},\ref{Peano},\ref{ATN}. Therefore,  we are seeking for a categorical axiomatisation of natural numbers. We can choose among: weak second order logic, logic of infinite disjunctions $\mathcal{L}_{\omega_1\omega}$ \cite{EEnge,CKarp} and algorithmic logic  $\mathcal{AL}$\cite{al:gm:as}. We prefer algorithmic logic for it has a system of axioms and inference rules. Moreover in the language of algorithmic logic we have ready at disposition  formulas that express the semantical properties of programs, such as termination, correctness, etc. \\
The next question which appears is: should we look for a new set of axioms of natural numbers or perhaps the set proposed by Peano will do. Consequently we shall consider three algorithmic theories and will find which of three is suitable for conducting the correctness proof of Euclid's algorithm. 
 The result of comparison we present in the table \ref{tabela1}.\\
\begin{table}[h]
\caption{Which theory allows to prove the halting formula \ref{H} of Euclid's  algorithm?}
\begin{center} 
\begin{tabular}{c}
\begin{tabular}{llp{2cm}p{5.5cm}}\hline
\textbf{Theory} & \textbf{Language}   \& \textbf{Logic} & \textbf{Axioms} & \textbf{Is there a proof?} \\ \hline \hline
$\mathcal{T}h_0$ & $\mathcal{L}$ -- 1-st order & Peano & \textit{No} - The formula does not appear in the language of first-order Peano Arithmetic.   There is nothing to prove. See also the stronger Fact \ref{hipo}. \\
$\mathcal{T}h_1$ & $\mathcal{L}_A$ -- algorithmic & Presburger & \textit{No} - the halting formula is independent from the axioms of this theory. A programmable counterexample is presented. See \ref{Findep1}.\\
$\mathcal{T}h_2$ & $\mathcal{L}_A$ -- algorithmic &  Peano & \textit{No} - the halting formula is independent from the axioms of this theory. See \ref{Findep2}.\\
$\mathcal{T}h_3$ & $\mathcal{L}_A$ -- algorithmic &  Algorithmic \newline Arithmetic & \textit{Yes} - there exists a proof. See theorem \ref{poprawny}. \\ \hline
\end{tabular}
\end{tabular} 
\end{center}\label{tabela1}
\end{table} 
The Euclid's algorithm is very important for mathematicians as well as for programmers. There is no doubt on it. However, the proofs of correctness of this algorithm do not satisfy us. Why? \\
One can split the goal of proving the correctness of an algorithm $A$ with respect to the given precondition $\alpha$ and postcondition $\beta$ onto two subgoals: 1) to prove that if some result exists then it satisfies the postcondition $\beta$, and 2) to prove that if the arguments satisfy the precondition $\alpha$ then the computation of the algorithm terminates. The first subgoal is easier. In the case of Euclid's algorithm, it suffices to remark that a common divisor of two numbers $n$ and $m$ is also a common divisor of $n$ and the difference $n-m$. From this the partial correctness of Euclid's algorithm w r.t. the pair of formulas ($n>0 \land m>0)$ and $(result = gcd(n,m))$  follows immediately. \\ 
The second subgoal is harder.   
We require that a proof of the correctness starts with axioms (either axioms of logic or axioms of natural numbers) and uses the inference rules of program calculus \textit{(i.e.} algorithmic logic) to deduce some intermediate formulas and to terminate with the halting formula.
All the proofs we know, do not satisfy this requirement. Let us take an example. In some monographs on theoretical arithmetic the proof goes as follow: 1\textdegree some intermediate formulas are proven, 2\textdegree a claim is made that   the scheme of induction is equivalent to the \textit{principle of minimum},  3\textdegree therefore for any natural numbers $n$ and $m$ the computation of Euclid's algorithm is finite and brings the $gcd(n,m)$. \\
Let us remark that in a proof like mentioned above :
\begin{itemize}
 \item [(\textit{i})] The claim on finiteness of descending sequences is true conditionally. The so called principle of minimum is valid only in the standard model. One has to assume that the algorithm works in the standard model of natural numbers. However, no elementary theory can guarantee that every of its models is isomorphic to the standard one. Moreover, the assumption on standard model is not written at all.
 \item [(\textit{ii})] The proof contains a phrase ``for any natural numbers $n $ and $m$ ...''. This is perfectly ok as long as the arguments of the algorithm are standard natural numbers. What happens if they come from another model of Presburger or Peano axioms? If Euclid's algorithm is executed in a non-standard model of Presburger arithmetic its computations may be infinite ones. See the Appendix A.
 \item[(\textit{iii})] The proof analyzes some sequences of numbers saying this is an execution  sequence of the algorithm. It means that the proof goes around. It would be preferable if the proof restrained from referring to semantics. 
\end{itemize}
\section{A short exposition of calculus of programs}\label{AL}
   
We shall discuss three algorithmic theories  $\mathcal{T}h_1,\mathcal{T}h_2,\mathcal{T}h_3$. All three theories have the same formalized language. All theories share the same consequence operation determined by the axioms and inference rules of program calculus, see appendix B. The theories have different sets of specific axioms.

All three algorithmic theories   $\mathcal{T}h_1,\mathcal{T}h_2,\mathcal{T}h_3$ have the same language $L=\langle A, \mathcal{WFF}\rangle $.
The alphabet  $A$ has the following subsets: set of functors $\varPhi =\{s, P, +, *, \stackrel{.}{\_} \}$, set of predicates $\varTheta = \{=, < \} $, set of logical operators $\{ \land, \lor, \implies, \lnot \}$, set of program operators \{:=, ;, \textbf{while, if} \}, and auxiliary symbols, parentheses and others.  The alphabet  $A$ contains also the set of variables.\\ 
Language  $L$ is an extension of the language of theory $\mathcal{T}h_0$, it admits programs and has a larger set of formulas. \smallskip\\
The reader familiar with the algorithmic logic \cite{al:gm:as} can safely skip the rest of this section. For the convenience of other readers we offer a few    words on the calculus of programs  and in the Appendix B we are listing  axioms and inference rules of the calculus.\\
A formalized logic $\mathcal{L}$ is determined by its language $L$ and the syntactic consequence operation $C$,  $\mathcal{L}=\langle L,C \rangle$. How to describe the difference between first-order logic FOL and algorithmic logic AL?
The language of algorithmic logic is a superset of the language of first-order logic and it is a superset of deterministic while programs, it includes algorithmic formulas and is closed by the  usual formation rules. In the language of AL we find all well formed expressions of FOL. The alphabets are similar. Moreover, the language of AL contains programs and the set of formulas is richer than the set of first-order formulas.  

As you can see the language $\mathcal{WFF}_{AL}$ contains programs. Moreover, the set of formulas $\mathcal{F}_{AL}$ is a proper superset of the set of first-order formulas $\mathcal{F}_{FOL} $. \smallskip  
\begin{center}\label{Rys1} 
\begin{tabular}{c} 
 \begin{tikzpicture}
 [scale=0.22]
 
\tikzstyle{logic}=[rectangle,draw,rounded corners=4pt]
\tikzstyle{next}=[->,>=latex,rounded corners=4pt]
\tikzstyle{aux}=[  ]
\node[logic] (P) at (20,4) {\begin{footnotesize}\begin{tabular}{l}propositional calculus PL\\\quad \ \ $\mathcal{WFF}_{PL}=\{ \mathcal{F}_{PL} \}$ \end{tabular}\end{footnotesize}};
\node[logic] (F) at (3,18) {\begin{footnotesize}\begin{tabular}{l}predicate calculus FOL\\\quad \ \ $\mathcal{WFF}_{FOL}=\{ \mathcal{T}_{FOL}\cup \mathcal{F}_{FOL}\}$ \end{tabular}\end{footnotesize}};
\node[logic] (A) at (20,34) {\begin{footnotesize}\begin{tabular}{l} program calculus AL\\\quad \ \ $\mathcal{WFF}_{AL}=\{ \mathcal{T}_{AL}\cup\mathcal{F}_{AL}\cup\mathcal{P}_{AL}\}$\\
$ \qquad \ \ \  
\mathcal{F}_{FOL}\varsubsetneq \mathcal{F}_{AL} $
\end{tabular}\end{footnotesize}};
\node[logic] (Y) at (37,18) {\begin{footnotesize}\begin{tabular}{l} calculus of program schemes PAL\\\quad \ \ $\mathcal{WFF}_{PAL}=\{ \mathcal{F}_{PAL}\cup \mathcal{P}_{PAL} \}$ \end{tabular}\end{footnotesize}};
\draw[next]  (P) -- (F) node[midway, right] { };
\draw[next]  (P) -- (Y) node[midway, right] { };
\draw[next]  (F) -- (A) node[midway, right] { };
\draw[next]  (Y) -- (A) node[midway, right] { }; 
\end{tikzpicture}   \\
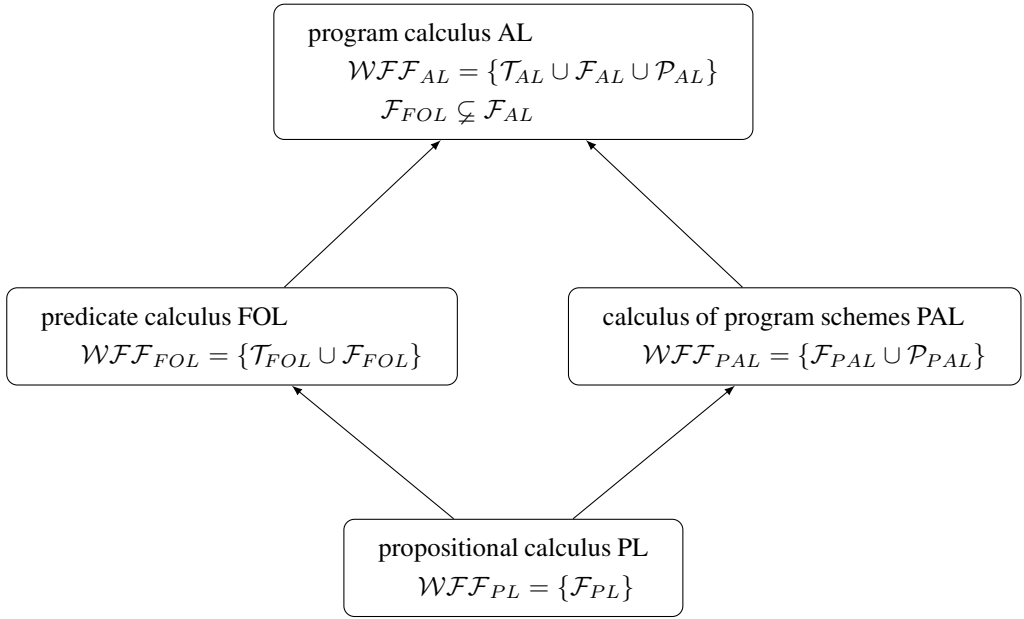
Figure 1. Comparison   of logical calculi w.r.t. their $\mathcal{WFF}$ sets
\end{tabular}
\end{center}

The set $\mathcal{WFF}_{AL}$ of well formed expressions is the union of three sets: set of terms (programmers may say, set of arithmetical expressions), set of formulas (i.e. set of boolean expressions) and the set of programs.
\begin{definition}The set of \textit{terms} is the least set of expressions   $T$ such that
\begin{itemize}
 \item each variable $x$ is  an element of the set  $T$,
 \item if an expression   $\tau$ belongs to the set   $T$, then the expressions   $s(\tau),\ P(\tau)$ belong to the set $T$,
 \item if expressions   $\tau$ and $\sigma$ belong to the set  $T$, then the expressions    $(\tau + \sigma), \ (\tau * \sigma), \ (\tau\stackrel{.}{\_}\sigma)$ belong to the set $T$.\hspace*{5.4cm} $\square$
\end{itemize}
\end{definition}
The set of formulas we describe in two steps. 
\begin{definition}
 The set of \textit{open formulas} is the least set $F_O$ of expressions such that 
 \begin{itemize}
  \item if expressions $\tau$ and $\sigma$ are terms, then the expressions  $(\tau=\sigma),\ (\tau<\sigma)$ are open formulas,
  \item if expressions $\alpha$ and $\beta$ are open formulas, then the expressions $(\alpha \land \beta)\ (\alpha \lor \beta),\ (\alpha \implies \beta),\ \lnot \alpha$ are open formulas.\hspace*{5.4cm} $\square$
 \end{itemize}
\end{definition}
\begin{definition}
The set of \textit{programs}  (in the language of theories $\mathcal{T}h_1,\mathcal{T}h_2, \mathcal{T}h_3$) is the least set $\mathcal{P}$ of expressions, such that  
\begin{itemize}
 \item If $x$ is a variable and an expression   $\tau$ is a term,   then the expression $x:=\tau$ is a program. (Programs of this form are called assignment instructions. They are atomic programs.)
 \item if expressions $K$ and $M$ are programs, then the expression $\{K;\,M\}$ is a program,
 \item if expression $\gamma$ is an open formula and expressions $K$ and $M$, are programs, then the expressions   $\textbf{while }\gamma \textbf{ do }M \textbf{ od}$ \qquad and \qquad $\textbf{if }\gamma \textbf{ then }K  \textbf{ else }M \textbf{ fi}$ \qquad are programs. \hspace*{5.4cm} $\square$
\end{itemize}
\end{definition}
We use the braces  \{ \} to delimit a program.
\begin{definition} 
 The set of \textit{formulas} is the least set of expressions  $F$ such, that
 \begin{itemize}
  \item each open formula belongs to the set   $F$,
  \item if an expression  $K$ is a program and an expression   $\alpha$ is a formula, then the expression   $K\,\alpha$ is a formula,
  \item if an expression $K$ is a program and an expression   $\alpha$ is a formula, then expressions   $\bigcup K\,\alpha$  and $\bigcap K\,\alpha$ are formulas,
  \item if an expression $\alpha$ is a formula, then the expressions  $\forall_x\,\alpha$ and $\exists_x\,\alpha$ are formulas,
  \item if expressions $\alpha$ and $\beta$ are formulas, then the expressions  $(\alpha \land \beta)\ (\alpha \lor \beta),\ (\alpha \Rightarrow \beta),\ \lnot \alpha$ are formulas.\hspace*{5.4cm} $\square$
 \end{itemize}
\end{definition}
Following Tarski we associate to each well formed expression of the language a mapping. The meanings of terms and open formulas is defined in a classical way. Semantics of programs requires the notion of computation (i.e. of execution). For the details consult \cite{al:gm:as}. Two facts would be helpful in reading further:\begin{itemize}                                                                                                                                                                                                                                                                                                                             \item The meaning of an algorithmic formula $K\alpha$ in a data structure $\mathfrak{A}$ is a function from the set of valuations of variables into two-element Boolean algebra $B_0$ defined as follow
\[(K\alpha)_{\mathfrak{A}}(v)= \begin{cases}
   \alpha_{\mathfrak{A}}(K_{\mathfrak{A}} (v))  & \mbox{if the result}\ K_{\mathfrak{A}}(v)\ \mbox{of computation  at initial valuation}\  v\ \mbox{is defined}  \\
   \textbf{false} & \mbox{otherwise i.e. if the computation of K fails or loops endlessly}
                               \end{cases}
 \]
 This explains why the formula \ref{H} expresses the halting property of the program \ref{E}. \\
 Define $K^i\alpha$ by induction: $K^0 \alpha = \alpha $ and $K^{i+1}\alpha=KK^i\alpha$. \\ 
 We read the formula      $\bigcup K\,\alpha$  as \textit{there exists an iteration of program} $K$ \textit{such that } $K^i\alpha$ \textit{holds}, and $\bigcap K\,\alpha$ means \textit{for each iteration of program} $K$ \textit{formula} $K^i\alpha$ \textit{holds}.\\ The signs $\bigcup $ and $\bigcap$ are \textit{iteration quantifiers.} The meaning of these formulas is defined as follow.
 \begin{align*} & (\bigcup K\alpha)_{\mathfrak{A}}(v)=  
   l.u.b.\ \{K^i\alpha_{\mathfrak{A}} (v)\}_{i \in N} \\ 
  & (\bigcap K\alpha)_{\mathfrak{A}}(v)=  
   g.l.b.\ \{K^i\alpha_{\mathfrak{A}} (v)\}_{i \in N} 
 \end{align*}
\item The calculus of programs i.e. algorithmic logic, enjoys the property of completeness. For the theorem on completeness consult \cite{al:gm:as}.                                                                                                                                                                                                                                                                                                                              \end{itemize}

\section{Elementary Peano arithmetic $\mathcal{T}h_0$}\label{firstPeano}
This first-order theory has been studied by many researchers. Many variants of the theory are known c.f. \cite{oml:ag}. The specific axioms of theory $\mathcal{T}h_0$ can be found in section  \ref{Peano}.

It is obvious that the language of the theory does not contain algorithms, and hence it does not contain algorithmic formulas. Therefore the halting formula (\ref{H}) of Euclid's algorithm can not be formulated nor proved in first-order theory of natural numbers with the axioms of Peano. 

Still, a problem arises, is there a   theorem of theory $\mathcal{T}h_0$, equivalent to the formula  (\ref{H})?\medskip \\
Let $\mathfrak{C}$    be a class of all data structures such with two-argument relation < and two-argument operation -.
\begin{pyta}\label{kwestia}\textit{ Is there a first-order formula $\psi(m,n)$, such that, for every data structure  $A$ in $\mathfrak{C}$, \\(i) it expresses the following property:}  the computation of Euclid's algorithm for arguments $m$ and $n$ is finite -- \textit{in this case formula $\psi(m,n)$ evaluates to \textbf{true},\\  (ii) }the formula $\psi(m,n)$ rejects the infinite computations in the  data structure,  -- \textit{i.e. it assumes the value \textbf{false} if the computation on arguments m, n can be prolonged} ad infinitum, \textit{and \\(iii) the formula $\forall_m\forall_n\,\psi(m,n)$ is a theorem of Peano arithmetic}.\end{pyta}
An answer to the question is brought by  Fact \ref{hipo}.
\section{Algorithmic theory of addition $\mathcal{T}h_1$}  \label{Th1}
We consider an algorithmic theory, henceforth its language contains programs and algorithmic formulas.  Note, all axioms of this theory are first-order formulas!  
\begin{definition}
The set of \textit{specific axioms} of the theory   $\mathcal{T}h_1$ consists of the following  formulas:
\begin{eqnarray}  
  &&\label{Aksj1}\forall_x\, s(x) \neq 0    \\
  &&\forall_x \forall_y\, s(x)=s(y) \Rightarrow x=y    \\
  && \forall_x\, x+0=x \\
  && \forall_x \forall_y\, x+s(y)=s(x+y) \\
  && x < y \Leftrightarrow \exists_z y=x+s(z) \\
  && P(0)=0 \\
  && P(s(x)) = x \\
  && z \stackrel{.}{\_} 0 = z \\
\label{Aksj9}  && z \stackrel{.}{\_} s(x) = P(z \stackrel{.}{\_} x) \\
 &&\nonumber \mbox{and an infinite set of formulas built in accordance with} \\
 && \nonumber \mbox{ the following scheme of induction:  }\\
\label{ind}  && \varPhi(x/0) \land \forall_x\bigl(\varPhi(x) \implies \varPhi(x/s(x))\bigr) \implies \forall_x\,\varPhi(x)   
\end{eqnarray} 
The last line is a scheme of infinitely many axioms.  It is the scheme of induction. The expression   $\varPhi$ denotes an arbitrary first-order formula with a free variable $x$. The expression  $\varPhi(x/0)$ denotes a formula resulting from the expression $\varPhi$ by the replacement of all free occurrences of variable $x$ by constant $0$.  Similarly, the expression  $\varPhi(x/s(x))$ is the formula that results from $\varPhi$ by the simultaneous replacement of all free occurrences of variable $x$ by the term  $s(x)$.\hspace*{5.4cm} $\square$
\end{definition}
Our set of axioms  differs insignificantly from those considered by Presburger. cf.  \cite{Presb,Stans}. \\
Observe, that the formula \ref{H} satisfies requirements \textit{(i)} and \textit{(ii)} of Question \ref{kwestia}.
\begin{fact}
 The formula  \ref{H} is not a theorem of the theory   $\mathcal{T}h_1$.
 \[ \mathcal{T}h_1 \nvdash \mathrm{\ref{H}} \]
\end{fact}
\begin{proof}
The formula \ref{H} is falsifiable in a non-standard model    $\mathfrak{N}$ of theory $\mathcal{T}h_1$, cf. Appendix A.  By completeness of algorithmic logic it follows that the formula is not a theorem of algorithmic theory  $\mathcal{T}h_1$.
\end{proof}
\begin{fact} \label{Findep1}
 The formula \ref{H} is independent of axioms  \ref{Aksj1} - \ref{ind}. 
\end{fact}
Remark that there exist a programmable model of this theory.
\section{Algorithmic Peano arithmetic -- theory  $\mathcal{T}h_2$} \label{Peano}
The language of this theory contains additional functor * of multiplication.
\begin{definition}
The set of axioms of the next theory  $\mathcal{T}h_2$ consists of formulas  \ref{Aksj1} - \ref{Aksj9} and two formulas \ref{aks11}, \ref{aks12} that define the operation of multiplication. Moreover, the set of axioms contains all the formulas built in accordance with scheme of induction,
\begin{eqnarray}
  && \label{aks11} \forall_x \, x*0= 0 \\
  && \label{aks12} \forall_x \forall_y\, x*s(y)=(x*y)+x \\
 && \nonumber \mbox{scheme\ of\ induction:  }\\
  && \nonumber \varPhi(x/0) \land \forall_x\bigl(\varPhi(x) \implies \varPhi(x/s(x))\bigr) \implies \forall_x\,\varPhi(x)   
\end{eqnarray}
As in the preceding section, we shall limit the scheme of induction: the formula $\varPhi(x)$ must be a first-order formula.\hspace*{5.4cm} $\square$ \end{definition}
Note, that all axioms of theory $\mathcal{T}h_2$ are first-order  formulas.
\begin{fact}
The theory $\mathcal{T}h_2$ has (at least) two non-isomorphic models. One is the standard model $\mathfrak{N}_0$ of theory $\mathcal{T}h_0$, another is a non-standard model $\mathfrak{N}$ of the same theory. \end{fact}
Let us remark that there is no recursive (i.e. programmable) non-standard model of the theory,c.f.\,\cite{Tenn}. \\
Despite the fact, that we extended the language adding the operator of multiplication and the set of axioms adding the definition of the operation of multiplication, the new theory does not contain a theorem on correctness of Euclid's algorithm. It is so because, in the non-standard model $\mathfrak{N}$ Euclid's algorithm has infinite computations for non-standard elements. \\
We recall, that the formula \ref{H} satisfies requirements \textit{(i)} and \textit{(ii)} of Question \ref{kwestia}.
\begin{fact}
 The formula  \ref{H} is not a theorem of the theory   $\mathcal{T}h_2$.
 \[ \mathcal{T}h_2 \nvdash \mathrm{\ref{H}} \] 
\end{fact}
\begin{fact}\label{Findep2}
 The formula \ref{H} is independent of axioms  \ref{Aksj1} - \ref{aks12}. 
\end{fact}
Now, we are able to answer the question \ref{kwestia} of section \ref{firstPeano}. Observe that the theory $\mathcal{T}h_2$ is richer and stronger that the theory $\mathcal{T}h_0$. It is richer for its set of $WFF$ is a proper superset of $WFF$ set of theory $\mathcal{T}h_0$. It is stronger for it is based on bigger set of logical axioms and inference rules. Both theories have the same set of specific axioms, namely the axioms of Peano arithmetic.  Both theories  share the same models.  
Note, if the richer theory $\mathcal{T}h_2$  does not contain a theorem on correctness of Euclid's algorithm, how can such a theorem appear in the poorer and weaker theory $\mathcal{T}h_0$? We shall write down this observation as the following 
\begin{fact}\label{hipo}
 Among theorems of first-order Peano arithmetic there is no formula that expresses the same meaning   as formula  \ref{H}. 
\end{fact}

\section{Algorithmic theory of natural numbers $\mathcal{T}h_3$} \label{ATN}
The set of specific axioms of the theory   $\mathcal{T}h_3$ contains the following formulas:  
\begin{align}
\tag{I}\label{I} & \color{black}{\forall_x\, s(x) \neq 0}  &  \\
\tag{M}\label{M} & \color{black}{\forall_x \forall_y\, s(x)=s(y) \Rightarrow x=y}  &  \\
\tag{S}\label{standard} & \color{black}{\forall_x\, \{y:=0; \textbf{while }y \neq x \textbf{ do }y:=s(y) \textbf{ od} \}(x=y)} &   \\
\nonumber & \makebox[13cm]{ \dotfill} &\\ 
\tag{A}\label{add} & x+y  \stackrel{df}{=}\left\{\begin{array}{l}t:=0; w:=x; \\ \textbf{while }t\neq y\textbf{ do }t:=s(t); w:=s(w) \textbf{ od}\end{array}\right\}w  &  \medskip\\
\tag{L}\label{less} & x<y \stackrel{df}{\equiv}\left\{\begin{array}{l} w:=0; \\ \textbf{while }w\neq y\land w\neq x \textbf{ do }\\ \quad w:=s(w) \\ \textbf{ od}\end{array}\right\}(w=x \land w\neq y) & \medskip\\
\tag{P}\label{predec}  & P(x)\stackrel{df}{=}\left\{\begin{array}{l}\begin{array}{l}w:=0;\\ 
\textbf{if }x \neq 0 \textbf{ then } \\ 
\quad \textbf{while } s(w)\neq x\textbf{ do }  w:=s(w) \textbf{ od} \\
\textbf{fi}\end{array} \end{array}\right\}w  & \medskip\\
\tag{O}\label{odejm} & x\stackrel{.}{\_}y \stackrel{df}{=}\left\{\begin{array}{l}w:=x; t:=0;\\ \mathbf{while }\ t\neq y\ \mathbf{ do }\ t:=s(t); w:=P(w)\ \mathbf{ od} \end{array}\right\}w &
\end{align}
The third of axioms  [\ref{standard}] is an algorithmic, (not a first-order), formula. It states that every element is a standard natural number.  \\
In addition to these three formulas [\ref{I},\,\ref{M},\,\ref{standard}] we assume four more axioms   [\ref{add},\,\ref{predec},\,\ref{odejm},\,\ref{less}] that are defining operations: \textit{addition, predecessor, subtraction}  (\textit{respectively} $+,\, P,\, \stackrel{.}{\_} $) and ordering relation  $<$. 
Axioms A, P, O that define operations are shorter versions of longer formulas, e.g.
$   x+y=z  \Leftrightarrow \{t:=0; w:=x; \textbf{while }t\neq y\textbf{ do }t:=s(t); w:=s(w) \textbf{ od}\}(w=z)$ . Observe that $ z=  \{t:=0; w:=x; \textbf{while }t\neq y\textbf{ do }t:=s(t); w:=s(w) \textbf{ od}\}z $ since the variable $z$ does not occur in this program and the program always halt. \\  
This theory $\mathcal{T}h_3$ contains a proof of the halting formula   (H) of the Euclid's algorithm (E). Why?   1\textdegree) The theory is categorical: \textit{every model $\mathfrak{M}$ of this theory is isomorphic to the standard model $\mathfrak{N}_0$} of natural numbers. Hence the theory is complete.  2\textdegree) Computations of Euclid's algorithm in the structure   $\mathfrak{N}_0$ are finite. This follows from the traditional proof of correctness of Euclid's algorithm. 3\textdegree)  Hence, the formula (H) is valid in each model of the theory   $\mathcal{T}h_3$. 4\textdegree) Therefore, there exists a proof of formula H and the formula   is a  theorem of the theory   $\mathcal{T}h_3$.  \\
We believe that it is instructive to write the proof of halting formula (\ref{H}). Below we are developing the proof \footnote{of the correctness of Euclid's algorithm in algorithmic theory of numbers}. We show its detailed (and formalizable) version.    
\begin{itemize}
\item We start by showing that all axioms of the theory   $\mathcal{T}h_2$ are theorems of the theory $\mathcal{T}h_3$. 
\item We shall prove also a couple of properties that occur in the traditional proof of Euclid's algorithm. 
 We omit the proofs of those properties of greatest common divisor, that are theorems of the theory   $\mathcal{T}h_0$. E.g. the formula  
$\bigl ((n>m \land m>0) \implies gcd(n,m)=gcd(n-m,m) \bigr )$ has a proof in the elementary theory of Peano. 
\item We deduce halting property of algorithm (\ref{E}) from the halting property of another program \ref{standard}.
\end{itemize}
\subsection{Scheme of induction}
 In this subsection we are  proving the scheme of induction from the axioms of algorithmic theory of natural numbers. 
\begin{lemma} The following formula is a theorem of the theory $\mathcal{T}h_3$.
 \[\mathcal{T}h_3 \vdash \color{black}{\{y:=0\}\bigcup\{y:=s(y)\}(x=y)}\]
\end{lemma}
\begin{proof}
The following equivalence is a theorem of algorithmic logic cf. \cite{al:gm:as} p. 62.
 \[\vdash \begin{array}{l} \{y:=0;  \mathbf{while}\,y \neq x \,\mathbf{do}\, y:=s(y) \, \mathbf{od}  \}(x=y)
  \Leftrightarrow  \\ \{y:=0\}\bigcup \{ \mathbf{if}\,y \neq x  \,\mathbf{then}\,y:=s(y)\,  \mathbf{ fi} \}(x=y)  \end{array}\]
 Another theorem of algorithmic logic is the following equivalence
 \[\vdash \begin{array}{l} \{y:=0\}\bigcup\{y:=s(y)\}(x=y)  \Leftrightarrow \\ \{y:=0\}\bigcup\{\textbf{if }y \neq x \textbf{ then }y:=s(y)\textbf{ fi}\}(x=y).\end{array}\]
 By propositional calculus we have
 \[\vdash \begin{array}{l}\{y:=0;  \textbf{while }y \neq x \textbf{ do }y:=s(y) \textbf{ od} \}(x=y) \Leftrightarrow \\ \{y:=0\}\bigcup\{y:=s(y)\}(x=y).\end{array}\]
 By modus ponens we obtain
 \[\mathcal{T}h_3 \vdash  \{y:=0\}\bigcup\{y:=s(y)\}(x=y).\]
\hspace*{12cm}\end{proof}
\begin{lemma}\label{lemat2}
The following equivalences are theorems of algorithmic logic. 
\begin{align}\nonumber & \vdash \color{black}{\{y:=0\}\bigcup\{y:=s(y)\}\alpha(y) \equiv \{x:=0\}\bigcup\{x:=s(x)\}\alpha(x)}\\
\nonumber & \vdash \color{black}{\{y:=0\}\bigcap\{y:=s(y)\}\alpha(y) \equiv \{x:=0\}\bigcap\{x:=s(x)\}\alpha(x)}\end{align}
\end{lemma}
\begin{proof}
Let $\alpha(x)$ be an arbitrary formula with free variable $x$. The expression  $\alpha(y)$ denotes the formula resulting from the formula   $\alpha(x)$ by the simultaneous replacement of all free occurrences of the variable  $x$ by the variable  $y$. It is easy to remark, that for every natural number  $i \in N$ the following formula is a theorem  
\[\vdash \alpha(y/s^i(0)) \Leftrightarrow \alpha(x/s^i(0)). \]
By  the axiom $Ax_{14} $ of the assignment instruction we obtain another fact , for every natural number  $i \in N$ the following formula is a theorem
\[\{y:=0\}\{y:=s(y)\}^i\,\alpha(y) \equiv \{x:=0\}\{x:=s(x)\}^i\,\alpha(x) \]
Now, we apply the axiom  $Ax_{23}$ and obtain, that for every natural number  $i $ the following formula is a theorem.
\[\{y:=0\}\{y:=s(y)\}^i\,\alpha(y) \implies \{x:=0\}\bigcup\{x:=s(x)\} \,\alpha(x) \]
We are ready to apply the rule $R_4$. We obtain the theorem
\[\{y:=0\}\bigcup\{y:=s(y)\} \,\alpha(y) \implies \{x:=0\}\bigcup\{x:=s(x)\} \,\alpha(x) .\]
In a similar manner we are proving the other implication and the formula 
\[ \{y:=0\}\bigcap\{y:=s(y)\} \,\alpha(y) \equiv \{x:=0\}\bigcap\{x:=s(x)\} \,\alpha(x) .\]
\end{proof}

In the proof of scheme of induction we shall use the following theorem.
\begin{metatwierdzenie}
For every formula  $\alpha$ the following formulas are theorems of algorithmic theory of natural numbers. 
\begin{eqnarray}
\label{dlakazdego} \mathcal{T}h_3 \vdash \color{black}{\forall_x\,\alpha(x)} & \color{black}{\Leftrightarrow} & \color{black}{\{x:=0\}\bigcap \{x:=s(x)\}\,\alpha(x)} \\
\label{istnieje}  \mathcal{T}h_3 \vdash \color{black}{\exists_x\,\alpha(x)} & \color{black}{\Leftrightarrow} & \color{black}{\{x:=0\}\bigcup \{x:=s(x)\}\,\alpha(x)}  
\end{eqnarray}
\end{metatwierdzenie}
\begin{proof}   
We shall prove the property   (\ref{istnieje}).
 Let $\alpha(x)$ be a formula.\\
Every formula of the following form is a theorem of algorithmic logic.
 \[\vdash \color{black}{\alpha(x) \implies \alpha(x) \land \{y:=0\}\bigcup\{y:=s(y)\}(x=y)}.\]
This leads to the following theorem of the theory $\mathcal{T}h_3$.
 \[\mathcal{T}h_3\vdash \color{black}{\alpha(x) \implies   \{y:=0\}\bigcup\{y:=s(y)\}(\alpha(x) \land x=y)}.\]
In the next step we obtain.
 \[\mathcal{T}h_3\vdash \color{black}{\alpha(x) \implies \{y:=0\}\bigcup\{y:=s(y)\}\alpha(y)}.\]
Now, we can introduce the existential quantifier into the antecedent of the implication (we use inference rule R6).
 \[\mathcal{T}h_3\vdash \color{black}{\exists_x\,\alpha(x) \implies   \{y:=0\}\bigcup\{y:=s(y)\}\alpha(y)}.\]
 By the previous lemma   \ref{lemat2} we obtain.
 \[\mathcal{T}h_3\vdash \color{black}{\exists_x\,\alpha(x) \implies   \{x:=0\}\bigcup\{x:=s(x)\}\alpha(x)}.\] 
The proof of other implication as well as of formula  (\ref{dlakazdego}) is left as an exercise.
\end{proof}
We are going to prove the scheme of induction.\\
\begin{metatwierdzenie}
Let   $\alpha(x)$ denote an arbitrary formula with a free variable $x$. The formula built in accordance with the following scheme is a theorem of algorithmic theory of natural numbers  $\mathcal{T}h_3$.
\begin{equation}
\mathcal{T}h_3 \vdash \color{black}{\Bigl(\alpha(x/0) \land \forall_x\bigl(\alpha(x) \Rightarrow\alpha(x/s(x))\bigr)\Bigr)\implies \forall_x \alpha(x)}
\end{equation}
\end{metatwierdzenie}
\begin{proof}
 In the expression below, $\beta$ denotes a formula, $K$ denotes a program. Each formula of the form  
 \[\vdash \color{black}{((\beta \land \bigcap K\,(\beta \implies K\beta))\implies \bigcap K\,\beta)} \]
 is a theorem of calculus of programs, i.e. algorithmic logic  (cf.\cite{al:gm:as} p.71(8)). \\
 Hence, every formula  of the following form is a theorem of algorithmic logic. 
 \[\vdash \color{black}{((\alpha(x) \land \bigcap \{x:=s(x)\}\,(\alpha(x) \implies \{x:=s(x)\}\alpha(x)))\implies \bigcap \{x:=s(x)\}\,\alpha(x))} \]
 We apply the auxiliary inference rule  $R2$
 \[\tag{R2}\color{black}{\frac{\alpha, K\,\textbf{true}}{K\,\alpha}} \]
  (it is easy to deduce this rule from the rule $R_2$) and obtain another theorem of AL
 \[\vdash \color{black}{\{x:=0\}((\alpha(x) \land \bigcap \{x:=s(x)\}\,(\alpha(x) \implies \{x:=s(x)\}\alpha(x)))\implies \bigcap \{x:=s(x)\}\,\alpha(x))} \]
 Assignment instruction distributes over conjunction (Ax15) and implication(cf. \cite{al:gm:as}p.70 formula (4)), hence   
 \[\vdash \color{black}{((\{x:=0\}\alpha(x) \land \{x:=0\}\bigcap \{x:=s(x)\}\,(\alpha(x) \implies \{x:=s(x)\}\alpha(x)))\implies \{x:=0\}\bigcap \{x:=s(x)\}\,\alpha(x))} \]
 We apply the axiom of assignment instruction 
 \[\vdash(\alpha(x/0) \land \{x:=0 \}\bigcap \{x:=s(x)\} \,(\alpha(x) \implies \alpha(x/s(x))))\implies  \{x:=0 \}\bigcap \{x:=s(x)\} \,\alpha(x)) \]
 Now, we use the fact that in the algorithmic theory of natural numbers the classical quantifiers and iteration quantifiers are mutually expressive. (cf. formula (\ref{dlakazdego}) ) 
 \[\vdash \color{black}{(\alpha(x/0) \land \underbrace{\{x:=0 \}\bigcap \{x:=s(x)\}}_{\forall_x}\,(\alpha(x) \implies \alpha(x/s(x))))\implies \underbrace{\{x:=0 \}\bigcap \{x:=s(x)\}}_{\forall_x}\,\alpha(x))} \]
 and obtain scheme of induction -- each formula of the following scheme is a theorem of algorithmic theory of natural numbers  $\mathcal{T}h_3$.
 \[\mathcal{T}h_3\vdash \color{black}{(\alpha(x/0) \land (\forall x) (\alpha(x) \implies \alpha(x/s(x))))\implies (\forall x) \alpha(x))} \]
\end{proof}
Observe the following useful property of natural numbers. Many proofs use the following lemma. 
\begin{lemma}\label{rown2} Let $\alpha$ be any formula. Any equivalence built in accordance to the following scheme is a theorem of theory $\mathcal{T}h_3$ 
\begin{eqnarray}\nonumber \mathcal{T}h_3 \vdash \color{black}{\left \{\begin{array}{l} t:=0; \\   \textbf{while } t\neq s(y) \\ \textbf{do }\\ \quad t:=s(t);\\     \textbf{od} \end{array} \right \}\alpha \Leftrightarrow 
\left \{\begin{array}{l} t:=0; \\   \textbf{while }t\neq y \\ \textbf{do }\\ \quad t:=s(t); \\     \textbf{od};\\ \textbf{if } t \neq s(y) \\ \textbf{then }t:=s(t); \\ \textbf{fi}\end{array}                                                                                                                                                               \right \} \alpha}.\end{eqnarray} 
\end{lemma}
\begin{proof}
The proof makes use of the following theorem of AL
\begin{eqnarray}\nonumber\vdash  \color{black}{\left \{\begin{array}{l} t:=0; \\   \textbf{while } t\neq s(y) \\ \textbf{do }\\ \quad t:=s(t);\\     \textbf{od} \end{array} \right \}\alpha \Leftrightarrow 
\left \{\begin{array}{l} t:=0; \\   \textbf{while }t\neq s(y) \\ \textbf{do }\\ \quad t:=s(t); \\     \textbf{od};\\ \textbf{if } t \neq s(y) \\ \textbf{then }t:=s(t); \\ \textbf{fi}\end{array}                                                                                                                                                               \right \} \alpha}.\end{eqnarray} 
\begin{flushright}
\begin{footnotesize}where $\alpha$ is any formula\end{footnotesize}\end{flushright} and the axiom  (\ref{M}). It suffices to consider the formulas   $\alpha$ of the form $\beta \land t=s(y)$ without loss of generality. 
\end{proof}
The lemma can be formulated in another way: the programs occurring in the lemma  \ref{rown2} are equivalent.
\subsection{Addition}
The operation of addition is defined in the theory  $\mathcal{T}h_3$ as follows.  
\begin{definition}
\[\tag{A} x+y=z \Leftrightarrow \left\{\begin{array}{l}t:=0; w:=x; \\ \textbf{while }t\neq y\textbf{ do }t:=s(t); w:=s(w) \textbf{ od}\end{array}\right\}(z=w) \]
\end{definition}
We have to check whether the definition is correct. It means that we should prove that for any pair of values $x, y$ there exists a result. Moreover, we should prove that the result is unique and that it satisfies the recursive equalities $x+0=x$ and $x+s(y)=s(x+y)$. First, we remark that for every     $x$ and  $y$ the result $w$ of addition is defined. 
\begin{lemma}
 \[\mathcal{T}h_3 \vdash \forall_x\,\forall_y\, \left\{\begin{array}{l}t:=0; w:=x; \\ \textbf{while }t\neq y\textbf{ do }t:=s(t); w:=s(w) \textbf{ od}\end{array}\right\}(t=y) \]
\end{lemma}
\begin{proof}
Proof starts with the axiom (\ref{standard}). Next, we use the following auxiliary inference rule of AL, cf. 
   \cite{al:gm:as}\,p. 73(19).  
 \[\color{black}{\frac{\{t:=0;\, \textbf{while }t\neq y\textbf{ do }t:=s(t) \textbf{ od}\}\alpha}{\{t:=0;\, \textbf{while }t\neq y\textbf{ do }t:=s(t); w:=s(w)\textbf{ od} \}\alpha}} \]
 Thus we proved the following theorem  
 \[\mathcal{T}h_3 \vdash  \{t:=0;\, \textbf{while }t\neq y\textbf{ do }t:=s(t); w:=s(w) \textbf{ od}\}(t=y) \]
 The latter formula can be preceded by the assignment instruction $w:=x$ (we use the inference rule R2). 
 \[\mathcal{T}h_3 \vdash  \color{black}{\left\{\begin{array}{l}w:=x;\,t:=0; \\ \textbf{while }t\neq y\textbf{ do }t:=s(t); w:=s(w) \textbf{ od}\end{array}\right\}(t=y)} \]
 In the next step we may interchange two assignment instructions, for they have no common variables. 
 \[\mathcal{T}h_3 \vdash  \color{black}{\left\{\begin{array}{l}t:=0;\,w:=x; \\ \textbf{while }t\neq y\textbf{ do }t:=s(t); w:=s(w) \textbf{ od}\end{array}\right\}(t=y)} \]
 Finally we can add the quantifiers (rule R7) and obtain the thesis of lemma.
 \[\mathcal{T}h_3 \vdash  \color{black}{\forall_x\,\forall_y\, \left\{\begin{array}{l}t:=0;\,w:=x; \\ \textbf{while }t\neq y\textbf{ do }t:=s(t); w:=s(w) \textbf{ od}\end{array}\right\}(t=y)}. \]
\hspace*{14cm}\end{proof}
Our next observation is
\begin{lemma}
\[\mathcal{T}h_3 \vdash  \color{black}{x+0=x} \] 
\end{lemma}
\begin{proof}
From the definition we have
\[ x+0 = z \Leftrightarrow \left\{\begin{array}{l}t:=0; w:=x; \\  \textbf{while }t\neq 0\textbf{ do }t:=s(t); w:=s(w) \textbf{ od}\end{array}\right\}(z=w) \]
We apply the axiom $Ax_{21}$ of while instruction to obtain
\[ x+0 =z \Leftrightarrow  \left\{\begin{array}{l}t:=0; w:=x; \\  \textbf{if }t=y \textbf{ then } \textbf{ else } \\ \textbf{while }t\neq 0\textbf{ do }t:=s(t); w:=s(w) \textbf{ od} \\ \mathbf{fi}\end{array}\right\}w \]
Indeed, from the properties of while instruction we obtain the implication.
\[\mathcal{T}h_3 \vdash \color{black}{y=0 \implies \left\{\begin{array}{l}t:=0; w:=x; \\ \textbf{while }t\neq y\textbf{ do } \\ \quad  t:=s(t); w:=s(w) \\ \mathbf{od} \end{array}\right\}\alpha \equiv \{t:=0; w:=x;\} \alpha}. \]
We conclude that $\mathcal{T}h_3 \vdash\color{black}{x+0=x}$.
\hspace*{6cm}\end{proof}
Our next goal is 
\begin{lemma}
\[\mathcal{T}h_3 \vdash  \color{black}{x+s(y)=s(x+y)} \] 
\end{lemma}
\begin{proof}
Proof uses the equivalence:
\begin{eqnarray}\nonumber \color{black}{\left \{\begin{array}{l} t:=0; \\ w:=x; \\ \textbf{while } t\neq s(y) \\ \textbf{do }\\ \quad t:=s(t);\\ \quad  w:=s(w) \\ \textbf{od} \end{array} \right \}\alpha \equiv 
\left \{\begin{array}{l} t:=0; \\ w:=x;\\ \textbf{while }t\neq y \\ \textbf{do }\\ \quad t:=s(t); \\ \quad w:=s(w) \\ \textbf{od};\\ \textbf{if } t \neq s(y) \\ \textbf{then }t:=s(t);w:=s(w); \\ \textbf{fi}\end{array}                                                                                                                                                               \right \} \alpha}.\end{eqnarray}
The expression $\alpha$ is any formula.  The equivalence is an instance of the lemma  \ref{rown2}.  
\hspace*{10cm}\end{proof}
\subsection{Ordering relation $<$}
\begin{definition}
\[\color{black}{x<y \stackrel{df}{=} \{ w:=0; \textbf{while }w\neq y\land w\neq x \textbf{ do } w:=s(w) \textbf{ od}\}(w=x \land w\neq y)}. \]
\end{definition}
We shall prove the useful property.  \\
\begin{lemma}\label{trichot}\[\mathcal{T}h_3 \vdash \color{black}{\forall_x \forall_y\,(x<y \lor x=y \lor y<x)}. \]
\end{lemma}
\begin{proof}
 It follows from the axiom  (\ref{standard}), that   
 \[\mathcal{T}h_3 \vdash \color{black}{\forall_x \forall_y\,\{ w:=0; \textbf{while }w\neq y\land w\neq x \textbf{ do } w:=s(w) \textbf{ od}\}(w=x \land w\neq y \lor w \neq x \land w=y \lor x=y)}. \]
 because, the formula   $\color{black}{\bigl((w=x \land w\neq y) \lor (w \neq x \land w=y) \lor x=y\bigr)}$ is a theorem of AL and the implication   $(w\neq y\land w\neq x) \implies w \neq x$ follows from axiom $Ax_6$.
 From the axiom of algorithmic logic   $Ax_{15}$ we deduce 
 \begin{eqnarray}\mathcal{T}h_3 \vdash \color{black}{\forall_x \forall_y\,\bigl (\{ w:=0; \textbf{while }w\neq y\land w\neq x \textbf{ do } w:=s(w) \textbf{ od}\}(w=x \land w\neq y)} \\ \color{black}{\lor \{ w:=0; \textbf{while }w\neq y\land w\neq x \textbf{ do } w:=s(w) \textbf{ od}\}(w \neq x \land w=y)} \\ \color{black}{\lor
  \{ w:=0; \textbf{while }w\neq y\land w\neq x \textbf{ do } w:=s(w) \textbf{ od}\}(x=y)\bigr)}.\end{eqnarray}
It is easy to observe, that the first and second line of the above formula are the definitions of  relations   $x<y$ and $y<x$. We can skip the program in the third line for   1\textdegree\, the program always terminates and   2\textdegree\, the values of variables $x$ and $y$ are not changed by the program\footnote{From the axiom \ref{standard} one can easily deduce the formula $\{w:=0; \textbf{while }w\neq y\land w\neq x \textbf{ do } w:=s(w) \textbf{ od}\}(w=y\lor w=x)$ and the following formula $(x=k)\implies \{w:=0; \textbf{while }w\neq y  \textbf{ do } w:=s(w) \textbf{ od}\}(x=k)$.}.  Finally we obtain \[\mathcal{T}h_3\vdash \forall_x \forall_y\,(x<y \lor x=y \lor y<x)\]. 
\end{proof}
\begin{lemma}
\[x < y \equiv \exists_z\,y = x + s(z)\]
\end{lemma}
\begin{proof}
We first recall the axiom (\ref{standard}) of the theory $\mathcal{T}h_3$
\[\mathcal{T}h_3 \vdash \{ w:=0; \textbf{while }w\neq y  \textbf{ do } w:=s(w) \textbf{ od}\}(w=y  ) .\]
From the definition of the predicate  $<$ we obtain
\[\mathcal{T}h_3 \vdash x<y \implies \{ w:=0; \textbf{while }w\neq y\land w\neq x \textbf{ do } w:=s(w) \textbf{ od}\}(w=x \land w\neq y) .\]
Therefore
\[\mathcal{T}h_3 \vdash x<y \implies \{ w:=x; \textbf{while }w\neq y  \textbf{ do } w:=s(w) \textbf{ od}\}(w=y  ) .\]
Since $x \neq y$, hence the assignment instruction \textsf{w:=s(w)} will be executed at least once. Speaking more precisely, the  former formula is equivalent to the following one by the axiom $Ax_{21}$.  
\[\mathcal{T}h_3 \vdash x<y \implies \{ w:=x;w:=s(w); \textbf{while }w\neq y  \textbf{ do } w:=s(w) \textbf{ od}\}(w=y  ) .\]
\end{proof}
\begin{lemma}
 \begin{align}& \mathcal{T}h_3 \vdash \forall_x \,x<s(x) \\ 
& \mathcal{T}h_3 \vdash(x<y) \equiv \{w:=x;\textbf{while }w\neq y \textbf{ do }w:=s(w) \textbf{ od} \}(w=y) \\ 
& \mathcal{T}h_3 \vdash \forall_x \forall_y\,x<y \implies x+z<y+z \end{align} \end{lemma}

\subsection{Predecessor}
The operation of predecessor is defined by the following axiom \ref{P}.
\begin{definition}
\[\tag{P}\label{P} P(x)\stackrel{df}{=}\left\{\begin{array}{l}w:=0;\\
\textbf{if }x \neq 0 \textbf{ then }\\ 
\quad \textbf{while } s(w)\neq x\textbf{ do }  w:=s(w) \textbf{ od} \\
\textbf{fi}\end{array}\right\}(w) \]
\end{definition}
\begin{lemma}\label{l5.10}
 \begin{align} & \mathcal{T}h_3 \vdash P(0)=0 \\ 
 & \mathcal{T}h_3 \vdash x\neq0 \implies s(P(x))=x \\ 
& \mathcal{T}h_3 \vdash x\neq0 \implies P(x)<x \\
& \mathcal{T}h_3 \vdash (x<y) \equiv \{w:=y;\textbf{while }w\neq x \textbf{ do }w:=P(w) \textbf{ od} \}(w=x) \\
  & \mathcal{T}h_3 \vdash P(s(x))=x \\
 \label{PoprzS}
& \mathrm{For\ every\ natural\ number}\  i\ \qquad \mathcal{T}h_3 \vdash  P^i(s^i(x)=x \end{align}
\end{lemma}
We shall prove the fundamental property of the predecessor operator.
\begin{theorem}\label{zmniej}
\[\mathcal{T}h_3 \vdash \forall_x\,\{\textbf{while }x\neq 0 \textbf{ do } x:=P(x) \textbf{ od} \}(x=0) \]
\end{theorem}
\begin{proof}
For every $i\in N$ the following formula is a theorem of AL
\[\forall_x\,\{y:=0;(\mathbf{if}\ y \neq x\ \mathbf{then}\ y:=s(y)\ \mathbf{fi})^i\}(x=y)  \implies \{y:=0;(\mathbf{if}\ y \neq x\ \mathbf{then}\ y:=s(y)\ \mathbf{fi})^i\}(x=y) . \]
We use the scheme of mathematical induction and the lemma   \ref{l5.10}(\ref{PoprzS}), to prove that for every   $i\in N$, the following formula is a theorem of the theory  $\mathcal{T}h_3$  
\[\forall_x\,\{y:=0;(\mathbf{if}\ y \neq x\ \mathbf{then}\ y:=s(y)\ \mathbf{fi})^i\}(x=y)  \implies \{(\mathbf{if}\ x \neq 0\ \mathbf{then}\ x:=P(x)\ \mathbf{fi})^i\}(x=0) . \]
Note, the antecedent in each implication asserts   $x=s^i(0)$, and the successor of the implication asserts  $0=P^i(x)$.\\
Hence we can apply the axiom  $Ax_{21}$ of AL and obtain that for every   $i \in N$ 
\[\mathcal{T}h_3 \vdash \forall_x\,\{y:=0;(\mathbf{if}\ y \neq x\ \mathbf{then}\ y:=s(y)\ \mathbf{fi})^i\}(x=y)  \implies \{\mathbf{while}\ x \neq 0\ \mathbf{do}\ x:=P(x)\ \mathbf{od}\}(x=0) . \]
Now, we apply the inference rule   $R_6$ to obtain
\[\mathcal{T}h_3 \vdash \forall_x\,\{y:=0;\mathbf{while}\ y \neq x\ \mathbf{do}\ y:=s(y)\ \mathbf{od}\}(x=y)  \implies \{\mathbf{while}\ x \neq 0\ \mathbf{do}\ x:=P(x)\ \mathbf{od}\}(x=0) . \]
The antecedent of this implication is the axiom   (\ref{standard}) of natural numbers. We deduce (by the rule $R_1$) \\
\hspace*{1cm}$\mathcal{T}h_3 \vdash \forall_x\, \{\mathbf{while}\, x \neq 0\, \mathbf{do}\, x:=P(x)\, \mathbf{od}\}(x=0).$  \hspace*{1.6cm}
\end{proof} 
\textbf{Remark}. This theorem states that Euclid's algorithm halts if one of its arguments is one. Below,we shall prove the halting property in general case.   \\
\textbf{Another remark}.  Every model $\mathfrak{M}$ of the theory   $\mathcal{T}h_1$ such that Euclid's algorithm halts when one of arguments is equal 1, is isomorphic to the standard model $\mathfrak{N}_0$ of natural numbers.  \\ 
\textbf{End of remarks}.\medskip\\
We need the following inference rule
\begin{lemma} \hfill \\
 Let  $\tau$ be a term such that no variable of a program $M$ occurs in it, $Var(\tau)\cap Var(M)=\emptyset$. If the formula $\bigl((x=\tau ) \implies M\, (x=P(\tau))\bigr)$ is a theorem of the theory   $\mathcal{T}h_3$, then the formula   \\ $\{\textbf{while } x \neq 0 \textbf{ do } M \textbf{ od}\}(x=0)$ is a theorem of the theory too. Hence,  the following inference rule is sound
 \[\mathcal{T}h_3\vdash\dfrac{ (x=\tau) \implies M\, (x=P(\tau))}{ \{\textbf{while } x \neq 0 \textbf{ do } M \textbf{ od}\}(x=0)}\]
 in the theory   $\mathcal{T}h_3$
\end{lemma}
\begin{proof}
For every $i$ the following formula is a theorem of AL
\[\{x:=P(x) \}^i(x=0) \implies \{M \}^i (x=0).\]
We are using the premise $((x=k)\implies \{M\}(x=P(k))$. Hence, for every   $i\in N$ 
\[\{x:=P(x) \}^i(x=0) \implies \{\textbf{while }x \neq 0 \textbf{ do }M\textbf{od} \}(x=0). \]
Now, we apply the rule  $R_3$ and obtain
\[\{\textbf{while }x\neq 0 \textbf{ do }x:=P(x)\textbf{ od} \}(x=0) \implies \{\textbf{while }x \neq 0 \textbf{ do }M\textbf{od} \}(x=0). \]
The antecedent of this implication has been proved earlier  (Th \ref{zmniej}), We apply the rule $R_1$ and finish the proof. 
\end{proof} 
The following lemma is useful in the proof of the main theorem.
\begin{lemma}\label{lemat19}
 The following inference rule is sound in the theory  $\mathcal{T}h_3$:
\[\mathcal{T}h_3\vdash\dfrac{ (x=k) \implies M\, (x<P(k))}{ \{\textbf{while } x \neq 0 \textbf{ do } M \textbf{ od}\}(x=0)}\] 
\end{lemma}
\begin{proof}
 The proof is similar to the proof of preceding lemma. We leave it as an exercise.
\end{proof}
 
\begin{corollary}\label{skoncz}
Let $x$ be an arbitrary number $x \in N$. Each descending sequence such that   $a_1=x$ and for every  $i,\  a_{i+1}<a_i$, is finite and contains at most   $x$ elements.
\end{corollary}

\subsection{Subtraction}
The operation of subtraction is defined by the following axiom \ref{odejm}.
\begin{definition}
\[\tag{O} x\stackrel{.}{\_}y\stackrel{df}{=}\{w:=x; t:=0; \mathbf{while }\ t\neq y\ \mathbf{ do }\ t:=s(t); w:=P(w)\ \mathbf{ od} \}(w) \]
\end{definition}
\begin{lemma}\label{6.13}\begin{align}
& \mathcal{T}h_3 \vdash \forall_x\, x\stackrel{.}{\_}0 = x  \\
& \mathcal{T}h_3 \vdash \forall_x \forall_y\, x\stackrel{.}{\_}s(y) = P(x\stackrel{.}{\_}y)\\
& \label{wmn}
\mathcal{T}h_3 \vdash \forall_x \forall_y\, (x > y > 0) \implies x\stackrel{.}{\_}y < x \\
& \mathcal{T}h_3 \vdash \forall_x \forall_y\, (x < y) \implies x\stackrel{.}{\_}y =0 \bigskip \\\end{align}
\end{lemma} 
\begin{remark}
Observe that we just proved that each theorem of algorithmic theory $\mathcal{T}h_2$ is also theorem of algorithmic theory $\mathcal{T}h_3$. 
\[\mathcal{T}h_2 \subsetneq \mathcal{T}h_3 \] 
\end{remark}
Theory $\mathcal{T}h_3$ is strictly richer than theory $\mathcal{T}h_2$.  
\section{Proof of correctness of Euclid's algorithm.}
The proof splits on two subgoals:
\begin{enumerate}
 \item[(\textit{i})] to prove that for any natural numbers $n$  and $m$, the computation of Euclid's algorithm is finite, \\
 i.e. we are to prove that the halting formula \ref{H} is a theorem of the theory $\mathcal{T}h_3$,
 \item[(\textit{ii})] to prove that the algorithm computes the greatest common divisor of numbers $n$ and $m$.
\end{enumerate}
It is rather easy to prove the following fact 
\begin{fact}
\begin{equation}\label{zakoncz}\mathcal{T}h_3\vdash  \left(\begin{array}{l}n \neq m\,\land \\ max(n,m)=p\end{array}\right)\implies\left\{\begin{array}{l}  \mathbf{if }\,n>m \\  \mathbf{then }\\ \quad  n:=n\stackrel{.}{\_}m \\  \mathbf{else } \\ \quad  m:=m\stackrel{.}{\_}n \\  \mathbf{fi}\\   \end{array}\right\} (max(n,m)<p)   \end{equation} 
\end{fact} 
\begin{proof}
In the proof we use the axiom Ax$_{20}$ of \textbf{if} instruction     and lemma \ref{6.13}, equation \ref{wmn}. 
\end{proof}
Now, by lemma  \ref{lemat19}, we obtain the desired formula \ref{H}. Hence the computations of Euclid's algorithm, in any structure that is a model of theory $\mathcal{T}h_3$, are finite.   \\
It remains to be proved 
\begin{fact}
\begin{equation}\label{niezm}\mathcal{T}h_3\vdash\Bigl ( (gcd(n,m)=p)\implies\left\{\begin{array}{l}  \mathbf{if }\,n>m \\  \mathbf{then }\\ \quad  n:=n\stackrel{.}{\_}m \\  \mathbf{else } \\ \quad  m:=m\stackrel{.}{\_}n \\  \mathbf{fi}\\   \end{array}\right\} (gcd(n,m)=p)\Bigr ). \end{equation} 
\end{fact}
In the proof we use a few useful facts
\[n>m \implies gcd(n,m)=gcd(n\stackrel{.}{\_}m,m) \]
\[m>n \implies gcd(n,m)=gcd(n,m\stackrel{.}{\_}n) \]
\[n=m \implies gcd(n,m)=n \]
All three implications are well known and we do not replicate their proofs here cf. 
 \cite{art:ag}. \\
Combining these observations with formulas  \ref{zakoncz} and \ref{niezm} we come to the conclusion that the formula expressing the correctness  of Euclid's algorithm is a theorem of theory $\mathcal{T}h_3$
\begin{theorem}\label{poprawny}
\begin{equation} \mathcal{T}h_3\vdash\left(\begin{array}{l}n_0>0 \land \\ m_0>0\end{array}\right)\implies\left\{\begin{array}{l}n:=n_0;\ m:=m_0; \\
\mathbf{while}\,n \neq  m\, \mathbf{ do }\\ \quad \mathbf{if }\,n>m \\ \quad \mathbf{then }\\ \qquad n:=n\stackrel{.}{\_}m \\ \quad \mathbf{else } \\ \qquad m:=m\stackrel{.}{\_}n \\ \quad \mathbf{fi}\\ \mathbf{ od} \end{array}\right\} (n=gcd(n_0,m_0)) \end{equation}
\end{theorem}
\begin{flushright}
This ends the proof.                    \end{flushright}
   
\subsection{Moreover} 
Euclid's algorithm has many equivalent forms. Consider the following algorithm
\[\left\{\begin{array}{l}
\textbf{while}\ n \neq m\ \textbf{do} \\
\quad \textbf{while}\ n>m\ \textbf{do}\ n\leftarrow n-m \ \textbf{od}; \\
\quad \textbf{while}\ m>n\ \textbf{do}\ m\leftarrow m-n \ \textbf{od}; \\
\textbf{od}
\end{array}
\right\} \]
Let $K$ and $M$ denote   programs, let $\alpha$ and $\beta$ denote open formulas, let $\varphi$ be any formula of calculus of programs. The following equivalence is a tautology of calculus of programs.
\[ \left\{\begin{array}{l}
\textbf{while}\ \alpha \  \textbf{do} \\
\quad \textbf{if}\ \beta\ \textbf{do}\ K \ \textbf{else}\ M \ \textbf{fi}  \\
\textbf{od}
\end{array}
\right\}\varphi  \Leftrightarrow
\left\{\begin{array}{l}
\textbf{while}\ \alpha\ \textbf{do} \\
\quad \textbf{while}\ \alpha\land\beta\ \textbf{do}\ K \ \textbf{od}; \\
\quad \textbf{while}\ \alpha\land\neg \beta\ \textbf{do}\ M \ \textbf{od}; \\
\textbf{od}
\end{array}
\right\}\varphi
\]
In our case $\beta: n>m$, $\alpha: n \neq m$. We can use the trichotomy property, see lemma \ref{trichot} and obtain  $\alpha \land \beta \equiv n>m$.\\
Now, the proof of halting property is easy. We apply yhe following theorem of algorithmic theory of natural numbers $\mathcal{T}h_3$
\[((n\stackrel{{.}}{\_}m)+(m\stackrel{.}{\_}n)=k)\implies\{\textbf{while}\ n>m\ \textbf{do}\ n:= n\stackrel{_{.}}{\_}m \ \textbf{od} \}((n\stackrel{.}{\_} m)+(m\stackrel{{.}}{\_}n)<k) \] 
A similar statement can be proved about the second instruction \textbf{while}. 
Now by axiom $Ax_{19}$ of composition of programs  we obtain 
\[((n\stackrel{{.}}{\_}m)+(m\stackrel{.}{\_}n)=k)\implies\left\{\begin{array}{l}\textbf{while}\ n>m\ \textbf{do}\ n:= n\stackrel{{.}}{\_}m \ \textbf{od}; \\
\textbf{while}\ m>n\ \textbf{do}\ m:= m\stackrel{{.}}{\_}n \ \textbf{od}
 \end{array}\right\}((n\stackrel{.}{\_} m)+(m\stackrel{{.}}{\_}n)<k) \]
One can apply lemma  \ref{lemat19} and arrive to 
\[ \left\{\begin{array}{l}
 \textbf{while}\ n\neq   m  \ \textbf{do} \\
\quad  \textbf{while}\ n>m\ \textbf{do}\ n:= n\stackrel{{.}}{\_}m \ \textbf{od}; \\
\quad \textbf{while}\ m>n\ \textbf{do}\ m:= m\stackrel{{.}}{\_}n \ \textbf{od}  \\
\textbf{od} 
 \end{array}\right\}(n = m) \]
\smallskip \\
\noindent -------------------------- \\ 
Making use of the lemma \ref{lemat19} we note another theorem of the theory $\mathcal{T}h_3$. It says that all computations of the following program are finite
\begin{theorem}
\[\mathcal{T}h_3\vdash \Bigl ( (n>m) \implies \left \{\begin{array}{l}
          r:=n; \\ 
          \mathbf{while }\, r \geq m \, \mathbf{ do } \\
          \quad r:=r\stackrel{.}{\_}m; \\
          \mathbf{od}
         \end{array}
 \right \} (0 \leq r<m ) \Bigr ) \]
\end{theorem}
This leads to another  
\begin{theorem}
\[\mathcal{T}h_3\vdash \Bigl ( (n>m) \implies \left \{\begin{array}{l}
          r:=n;   q:=0; \\
          \mathbf{while }\, r \geq m \, \mathbf{ do } \\
          \quad r:=r\stackrel{.}{\_}m; \\
         \quad q:=q+1 \\
          \mathbf{od}
         \end{array}
 \right \} (0 \leq r<m \land n=q*m+r) \Bigr ) \]\end{theorem}
And the following  \begin{theorem}
 \[\mathcal{T}h_3\vdash \Bigl ( (n_0>m_0) \implies \left \{\begin{array}{l}
          n:=n_0;\ m:=m_0;\ r:=n;     \\
          \mathbf{while }\, r \neq 0 \, \mathbf{ do } \\
          \quad r:=n; \\
          \quad \mathbf{while }\, r \geq m \,\mathbf{ do} \\
          \qquad r :=r \stackrel{.}{\_}m \\
          \quad \mathbf{od}; \\
          \quad n:=m; \\
          \quad m:=r \\
          \mathbf{od}
         \end{array}
 \right \} (n=gcd(n_0,m_0) ) \Bigr ) \] \end{theorem}
\section{Final remarks}
So far, we succeeded in developing one small chapter of algorithmic theory of natural numbers. The whole theory contains much more theorems. Some are first-order formulas, some are algorithmic formulas. The theorem on correctness of Euclid's algorithm is deduced from a couple of earlier theorems. The algorithmic theory of numbers does not begin nor does it end by this theorem. We claim that the calculus of programs (i.e. algorithmic logic) is a useful tool in building the algorithmic theory of numbers. Note the Fact \ref{hipo}. Think of its consequences. \smallskip \\
One has to take into consideration that the future development of algorithmic theory of numbers will demand to analyze more complicated algorithms -- the intuitive way of describing computations may happen to be  error prone or leading to paradoxes. On the other hand, it is very probable that, in programming, we shall encounter some erroneous (or fake) classes that pretend to implement the structure of unsigned integer (i.e. natural numbers)\footnote{In computer algebra and in computer geometry one may need Euclid's algorithm executing in a class that implements the algebra of polynomials or structure of segments in a three dimensional space. It is of importance to check that such a class is not a non-standard model similar to our class $Cn$ of Appendix A.}. \\ 
We hope, that programmers and computer scientists will note that \textit{proving of programs} need not to start a new, with every program one wishes to analyze. In the process of proving some semantical property $s_P$ of a certain program $P$, one can use lemmas and theorems on other semantical properties of programs, that have been proved earlier. We demonstrate this pattern within the correctness proof of Euclid's algorithm. In other words, we propose to develop the algorithmic theory of natural numbers. In fact, we did it in the book \cite{al:gm:as} p.\,155. Such a theory may be of interest also to mathematicians. One can note the appearance of books on algorithmic theory of numbers \cite{atn:bach, atn:lovasz}, algorithmic theory of graphs\cite{agt:mchugh, agt:gibbons}, etc. We are offering calculus of programs i.e. algorithmic logic as a tool helpful in everyday work of informaticians and mathematicians. 
\subsection*{Acknowledgments}
Grażyna Mirkowska read the manuscript and made several helpful suggestions. \medskip \\
I wrote this paper in 2015 and submitted it to Fundamenta Informaticae. I have got two reviews. One positive, critical and helpful. Anotherone, was aggresive and negative. After exchange of letters through the editor Damian Niwiński, the referee changed his opinion and suggested to publish  the paper. Then ... nothing. \\
I submitted the paper to Information and Computation. After 4 years of awaiting I have obtained a negative decision. The refereee said that he/she is canfused -- because it is not Hoare logic paper.  There was no evaluation of contents of the submitted paper.\\
The present author will appreciate any comments.
\section{Suplements}
\subsection{A  class implementing nonstandard model of theory $\mathcal{T}h_1$}
In this section we present a class \textit{Cn} that implements a programmable and non-standard model $\mathfrak{M}$ of axioms of addition theory   $\mathcal{T}h_1$ (cf. section. \ref{Th1}). We show that Euclid's algorithm, executed in this model has infinite computations.   \\

It is well known that the set of axioms of the theory $\mathcal{T}h_1$ has non-standard models. We are reminding that the system
\[\mathfrak{M}=\langle M, zero, one, s, add, subtract;equal, less \rangle \]
where  
\begin{itemize}
 \item The set $M$ is defined as follow  
\[\langle k, x \rangle \in M \equiv \{k \in \mathcal{Z} \wedge x \in \mathcal{R} \wedge x \geq 0 \wedge (x=0 \implies k \geq 0)\} \]
here  $k$ is an integer, $x$ is a non-negative rational number and when $x$ is 0 then $k \geq 0$ ,  
\item the operation addition is defined component wise, as usual in a product,
\item the successor operation is defined as follow $s(\langle k, x \rangle)=\langle k+1, x \rangle $,
\item constant zero 0 is $\langle 0,0 \rangle $.
\item relation \textit{less} has a type $\omega + (\omega^*+\omega)\cdot\eta$
\end{itemize}
is a non-standard and recursive(i.e. computable) model of axioms of theory $\mathcal{T}h_1$.

Now, class \textit{Cn} is written in Loglan programming language \cite{logl82}. This class defines and implements an algebraic structure $\mathfrak{C}$. The universe of the structure consists of all objects of the class $NCN$ (this is an infinite set). Operations in the structure $\mathfrak{C}$ are defined by the methods of  class $Cn$: \textit{add, equal, zero} and \textit{s}. All the axioms of the algorithmic theory $\mathcal{T}h_1$ are valid in the structure $\mathfrak{C}$, i.e. the structure is a model of the theory. We show that for some data the execution of  Euclid's algorithm is infinite.\smallskip \\
 \begin{theorem}
 	The algebraic structure $\mathfrak{C}$ which consists of the set $|NSN|$ of all objects of class \textit{NSN} together with the methods \textit{add}, \textit{s}, \textit{equal} and constant \textit{zero}, \textit{equal, \textit{less}} 
 	\[\mathfrak{C}= \langle |NSN|,zero, s, add, subtract, equal, less  \rangle \]
 	satisfies all axioms of natural numbers with addition operation, cf. section \ref{Th1}. 
 \end{theorem}
 \begin{proof} This is a slight modification of the arguments found in Grzegorczyk's book \cite{art:ag}p.239.  \end{proof}
\begin{footnotesize}
\quad\begin{tabular}{l} \hline
\textbf{unit}\ Cn:   \ \textbf{class}; \\

\quad $\left [\begin{array}{l}\textbf{unit} \  NSN:\     \textbf{class} (intpart,nomprt, denom: integer); \\
  \textbf{begin} \\
 \ \ \  \textbf{if} \ nomprt= 0  \ \textbf{and} \ intpart <0\,  \textbf{orif}   \ nomprt*denom < 0\ \textbf{orif}    
  \ denom =0 \\
  \ \ \   \textbf{then \ raise} \ Exception \ \textbf{fi} \\
  \textbf{end} \ NSN; 
\end{array}\right .$\smallskip \\
\quad $\left [\begin{array}{l}\textbf{unit} \  add:  \textbf{function} (n,m: NSN) : NSN; \\
  \textbf{begin} \     result:=  \textbf{new}\  NSN(n.intpart+m.intpart, \\ \ \ \    n.nomprt*m.denom+n.denom*m.nomprt, n.denom*m.denom)\    \textbf{end} \ add; \end{array}\right .$\smallskip \\
\quad $\left [\begin{array}{l}\textbf{unit} \  subtract: \ \textbf{function} (n,m: NSN) : NSN; \\
  \textbf{begin} \  \textbf{if}\ less(n,m)\ \textbf{then}\ result:= zero \\   \ \ \ \ \   \textbf{else}\ result := \textbf{new} NSN(n.intprt - m.intprt,\\ \ \ \ \ \ \ \ \  n.nomprt * m.denom - n.denom * m.nomprt, n.denom * m.denom)\ \textbf{fi}  \\  \textbf{end} \ subtract; \end{array}\right .$\smallskip \\
\quad $\left [\begin{array}{l}\textbf{unit} \ equal:  \textbf{function} (n,m: NSN): Boolean; \\
  \textbf{begin} \  \quad  result := (n.intpart=m.intpart)\ \textbf{and}\\ \ \ \  (n.nomprt*m.denom=n.denom*m.nomprt)\   \textbf{end} \ equal; \end{array}\right .$\smallskip \\
\quad $\left [\begin{array}{l}\textbf{unit} \ zero:  \textbf{function}: NSN; \\
  \textbf{begin} \  result := \textbf{new}\ NSN(0,0,1)\   \textbf{end} \ zero; \end{array}\right .$\smallskip\\
\quad $\left [\begin{array}{l}\textbf{unit} \ s:  \textbf{function}(n: NSN): NSN; \\
  \textbf{begin} \  result := \textbf{new}\ NSN(n.intpart +1, n.nomprt, n.denom)  \textbf{ end} \  s; \end{array}\right .$\smallskip \\
\quad $\left [\begin{array}{l} \textbf{unit}\ less:\  \textbf{function} (n,m: NSN) : Boolean; \\
  \textbf{begin} \  \textbf{if}\ n.nomprt=0\ \textbf{andif}\ m.nomprt=0\ \textbf{then}\ result := n.intprt < m.intprt \\   \ \ \ 
 \textbf{else}\  \textbf{if}\  n.nomprt=0\ \textbf{andif}\ m.nomprt >0\ \textbf{then}\   result:= true \\   \ \ \ \ \textbf{else}\ \textbf{if}\  n.nomprt>0\ \textbf{andif}\ m.nomprt =0\ \textbf{then}\ result := false \\   \ \ \ \ \ \ \textbf{else}\  \textbf{if}\  n.intprt =/=  m.intprt\ \textbf{then}\ result := n.intprt < m. intprt\\   \ \ \ \ \ \ \ \  \textbf{else}\ result :=  n.nomprt*m.denom<n.denom*m.nomprt\ \textbf{fi} \ \textbf{fi}\  \textbf{fi}\ \textbf{fi}\  \textbf{end}\  less \end{array}\right .$\smallskip \\
\textbf{end} \ Cn;\\ \hline
 \end{tabular} \end{footnotesize}  \bigskip \\

\noindent Have a look at the following example and verify  that Euclid's algorithm  has infinite computations, i.e. does not halt, when interpreted in the data structure $\mathfrak{C}$.

\begin{example}
Suppose that the values of variables $x,y,z$ are determined by the execution of three instructions  \\
\begin{footnotesize}\hspace*{0.3cm} $x:= \textbf{new}\ NSN(12,0,1);$ \ \
\hspace*{0.3cm} $y := \textbf{new}\ NSN(15, 0,2);$ \ \
\hspace*{0.3cm} $z := \textbf{new}\ NSN(15, 1,2);$                      \end{footnotesize}\\
Now, the computation of the algorithm  $E(m,n)$:
{\footnotesize \[ \left\{ \begin{array}{l}\mathbf{while\ not}\ equal(m,n)\  \mathbf{do} \\ \quad \mathbf{if}\ less(m,n) \\ \quad \mathbf{then}\ n:=subtract(n,m) \\ \quad \mathbf{else}\ n:=subtract(n,m) \\ \quad  \mathbf{fi} \\ \mathbf{od} \end{array}\right\}                                                                                                                                                                 \] }
 for $m=x$ and $n=y$ is finite and results is \begin{footnotesize}$\textbf{new}\ NSN(3,0,1) $                                                                                                    \end{footnotesize}.\smallskip\\
An attempt to compute $E(x,z)$ results in an infinite computation, or more precisely, in a computation that can be arbitrarily prolonged, as it is shown in the table below. \\
\begin{center}
\begin{footnotesize}\begin{tabular}{l}
States of memory during a computation\\
\begin{tabular}{ll} \hline 
\textbf{n} & \textbf{m}   \\ \hline \hline
 \textbf{new}\ NSN(12,0,1) & \textbf{new}\ NSN(15, 1,2) \\
 \textbf{new}\ NSN(12,0,1) & \textbf{new}\ NSN(3, 1,2) \\
 \textbf{new}\ NSN(12,0,1) & \textbf{new}\ NSN(-9, 1,2) \\
 \textbf{new}\ NSN(12,0,1) & \textbf{new}\ NSN(-21, 1,2) \\
 \textbf{new}\ NSN(12,0,1) & \textbf{new}\ NSN(-33, 1,2) \\
 ... & ... \\
 \textbf{new}\ NSN(12,0,1) & \textbf{new}\ NSN(15-i*12, 1,2) \\
 ... & ...   \\ \hline
\end{tabular}             
\end{tabular}             \end{footnotesize}\end{center}
\end{example}   
Class $Cn$ which implements a non-standard programmable model  of theory $\mathcal{T}h_1$ has more applications.
\begin{itemize}
	\item  An application of class  $ Cn $ leads to the construction of non-standard model of   elementary theory of stacks and to  decidability of the theory , cf. \cite{MSST:path}
	\item We used the  class to show that the Collatz conjecture is not expressible by any first-order formula and is independent from elementary theory of natural numbers,  c.f. \cite{CCbT}
\end{itemize}

\subsection{Axioms and inference rules of program calculus $ \mathcal{AL} $}
For the convenience of reader we cite the axioms and inference rules of algorithmic logic.\\
\textbf{Note}. \textit{Every axiom of algorihmic logic is a  tautology. \\
Every inference rule of AL is sound.} \cite{al:gm:as}\bigskip\\
\begin{large}\textbf{Axioms}      \end{large}  
\begin{trivlist}
\color{black}{
\item[] \textit{axioms of propositional calculus}
\item[$Ax_1$] $((\alpha\Rightarrow\beta
)\Rightarrow((\beta\Rightarrow\delta)\Rightarrow(\alpha
\Rightarrow\delta)))$
\item[$Ax_2$] $(\alpha\Rightarrow(\alpha
\vee\beta))$
\item[$Ax_3$] $(\beta\Rightarrow(\alpha\vee\beta))$%
\item[$Ax_4$] $((\alpha\Rightarrow\delta)~\Rightarrow((\beta
\Rightarrow\delta)~\Rightarrow~((\alpha\vee\beta
)\Rightarrow\delta)))$
\item[$Ax_5$] $((\alpha\wedge\beta
)\Rightarrow\alpha)$
\item[$Ax_6$] $((\alpha\wedge\beta
)\Rightarrow\beta)$
\item[$Ax_7$] $((\delta\Rightarrow\alpha
)\Rightarrow((\delta\Rightarrow\beta)\Rightarrow(\delta
\Rightarrow(\alpha\wedge\beta))))$
\item[$Ax_8$] $((\alpha
\Rightarrow(\beta\Rightarrow\delta))\Leftrightarrow((\alpha\wedge\beta
)\Rightarrow\delta))$
\item[$Ax_9$] $((\alpha\wedge\lnot\alpha
)\Rightarrow\beta)$
\item[$Ax_{10}$] $((\alpha\Rightarrow
(\alpha\wedge\lnot\alpha))\Rightarrow\lnot\alpha)$
\item[$Ax_{11}$] $(\alpha\vee\lnot\alpha)$ }
\color{black}{
\item[] \textit{axioms of predicate calculus}
\item[$Ax_{12}$] $ ((\forall x)\alpha(x)\Rightarrow \alpha(x/\tau)))$ \ \ \newline
\hspace*{1.5cm}  {\small where term }$\tau${\small\ is of the same type as
the variable x}
\item[$Ax_{13}$] $(\forall x)\alpha(x)\Leftrightarrow\lnot
(\exists x)\lnot\alpha(x)$ }
\color{black} {
\item[] \textit{axioms of calculus of programs} 
\item[$Ax_{14}$] $K((\exists x)\alpha
(x))\Leftrightarrow(\exists y)(K\alpha(x/y))$ \begin{small}for $y\notin V(K)$                                                       \end{small}
\item[$Ax_{15}$] $K(\alpha\vee\beta)\Leftrightarrow((K\alpha)\vee(K\beta))$
\item[$Ax_{16}$] $K(\alpha\wedge\beta)\Leftrightarrow((K\alpha)\wedge(K\beta))$%
\item[$Ax_{17}$] $K(\lnot\alpha)\Rightarrow\lnot(K\alpha)$%
\item[$Ax_{18}$] $((x:=\tau)\gamma\Leftrightarrow(\gamma(x/\tau)\wedge
(x:=\tau)true))~\wedge((q:=\gamma\prime)\gamma\Leftrightarrow\gamma(q/\gamma
\prime))$
\item[$Ax_{19}$] \textbf{begin} $K;M$ \textbf{end} $\alpha
\Leftrightarrow K(M\alpha)$
\item[$Ax_{20}$] \textbf{if} $\gamma$ \textbf{then} $K
$ \textbf{else} $M$ \textbf{fi} $\alpha\Leftrightarrow((\lnot\gamma\wedge M\alpha
)\vee(\gamma\wedge K\alpha))$
\item[$Ax_{21}$] \textbf{while}~
$\gamma$~\textbf{do}~$K$~\textbf{od}~$\alpha\Leftrightarrow((\lnot\gamma\wedge\alpha
)\vee(\gamma\wedge K($\textbf{while}~$\gamma$~\textbf{do}~$K$~\textbf{od}$%
(\lnot\gamma\wedge\alpha))))$
\item[$Ax_{22}$] $\displaystyle{\bigcap K\alpha\Leftrightarrow
(\alpha\wedge(K\bigcap K\alpha))}$
\item[$Ax_{23}$] $\displaystyle{\bigcup K\alpha
\equiv(\alpha\vee(K\bigcup K\alpha))}$ }
\end{trivlist}
\begin{large}\textbf{Inference rules}\end{large} \\
\begin{trivlist}
\color{black}{
\item[] \textit{propositional calculus}
\item[$R_1$]\qquad $\dfrac{\alpha ,(\alpha \Rightarrow \beta )}{\beta }$ } \qquad (also known as modus ponens)
\color{black} {
\item[] \textit{predicate calculus}
\item[$R_{6}$]\qquad $\dfrac{(\alpha (x)~\Rightarrow ~\beta )}{((\exists
x)\alpha (x)~\Rightarrow ~\beta )}$
\item[$R_{7}$]\qquad $\dfrac{(\beta ~\Rightarrow ~\alpha (x))}{(\beta
\Rightarrow (\forall x)\alpha (x))}$ }
\color{black} {
\item[]  \textit{calculus of programs AL}
\item[$R_{2}$]\qquad $\dfrac{
(\alpha \Rightarrow \beta )}{(K\alpha \Rightarrow K\beta )
}$
\item[$R_{3}$]\qquad  $\dfrac{\{s({\bf if}\ \gamma \ {\bf then}\ K\ {\bf fi})^{i}(\lnot
\gamma \wedge \alpha )\Rightarrow \beta \}_{i\in N}}{(s({\bf while}\
\gamma \ {\bf do}\ K\ {\bf od}\ \alpha )\Rightarrow \beta )}$
\item[$R_{4}$]\qquad  $\dfrac{\{(K^i\alpha \Rightarrow \beta )\}_{i\in N}}{(\bigcup
K\alpha \Rightarrow \beta )}$
\item[$R_{5}$]\qquad $\dfrac{\{(\alpha \Rightarrow K^i\beta )\}_{i\in N}}{(\alpha
\Rightarrow \bigcap K\beta )}$ }
\end{trivlist}

In rules $R_6$ and $R_7$, it is assumed that $x$ is a variable which is not free in $\beta $, i.e. $x \notin  FV(\beta )$. The rules are known as the rule for
introducing an existential quantifier into the antecedent of an implication
and the rule for introducing a universal quantifier into the suc\-ces\-sor
of an implication. The rules $R_4$ and $R_5$ are algorithmic counterparts of rules
$R_6$ and $R_7$. They are of a different character, however, since their sets of
premises are in\-finite. The rule $R_3$ for introducing a \textbf{while} into the
antecedent of an implication of a similar nature. These three rules are
called $\omega $-rules.
The rule $R_{1}$ is known as \textit{modus ponens}, or the \textit{cut}-rule.

In all the above schemes of axioms and inference rules, $\alpha $, $\beta $, 
$\delta $ are arbi\-trary for\-mulas, $\gamma $ and $\gamma ^{\prime }$ are
arbitrary open formulas, $\tau $ is an arbitrary term, $s$ is a finite
se\-quence of assignment instructions, and $K$ and $M$ are arbitrary
programs. \smallskip \\


\end{document}